\documentclass[10pt]{llncs}
\usepackage{amscd,amsfonts,amsmath,amssymb}%
\usepackage{mathrsfs}
\usepackage{latexsym}%indentfirst,
\usepackage{multicol}
\usepackage{graphicx}
\usepackage{epsfig}
\usepackage{color}
\usepackage{ifthen}
\usepackage{algorithm}
\usepackage{algorithmic}
\usepackage[all]{xy}
\newcommand{\KeyGen}{\mathtt{KeyGen}}
\newcommand{\Dec}{\mathtt{Dec}}
\newcommand{\Enc}{\mathtt{Enc}}
\newcommand{\A}{\mathcal{A}}
\newcommand{\B}{\mathcal{B}}
\begin{document}
\frontmatter          % for the preliminaries
\mainmatter              % start of the contributions
\title{Public Key Encryption in Non-Abelian Groups}
\titlerunning{A New Miniature CCA2 Public key Encryption scheme}

\author{Haibo Hong$^1$\thanks{Corresponding Author: honghaibo1985@163.com}, Jun Shao$^1$, Licheng Wang$^2$, Haseeb Ahmad$^2$ \and Yixian Yang$^2$ }
\authorrunning{H.~Hong, L.~Wang, J.~Shao, H.~Ahmad and Y.~Yang}
\institute{$^1$School of Computer Science and Information Engineering, Zhejiang Gongshang University, Hangzhou, 3100018 P.R. China  \\
$^2$Information Security Center, State Key Laboratory of Networking and Switching Technology, Beijing University of Posts and Telecommunications, Beijing, 100876 P.R. China}

\maketitle              % typeset the title of the contribution

\begin{abstract}
In this paper, we propose a brand new public key encryption scheme in the Lie group that is a non-abelian group. In particular, we firstly investigate the intractability assumptions in the Lie group, including the non-abelian factoring assumption and non-abelian inserting assumption. After that, by using the FO technique, a CCA secure public key encryption scheme in the Lie group is proposed. At last, we present the security proof in the random oracle based on the non-abelian inserting assumption.
%
%
%Since 1870s, scientists have been taking deep insight into Lie groups and Lie algebras. With the development of Lie theory, Lie groups have got profound significance in many branches of mathematics and physics. In Lie theory, exponential mapping between Lie groups and Lie algebras plays a crucial role. Exponential mapping is the mechanism for passing information from Lie algebras to Lie groups. Since many computations are performed much more easily by employing Lie algebras, exponential mapping is indispensable while studying Lie groups. In this paper, we first put forward a novel idea of designing cryptosystem based on Lie groups and Lie algebras. Besides, combing with discrete logarithm problem(DLP) and factorization problem(FP), we propose some new intractable assumptions based on exponential mapping. Moreover, in analog with Boyen's sceme(AsiaCrypt 2007), we disign a public key encryption scheme based on non-abelian factorization problems in Lie Groups. Finally, our proposal is proved to be IND-CCA2 secure in the random oracle model.
%\\
%
%
% %ÆäÀíÂÛÓë·½·¨ÒÑÉø͸µ½ÊýѧºÍÀíÂÛÎïÀíµÄÐí¶àÁìÓò¡£ÔÚ±¾ÎÄÖУ¬ÎÒÃÇÊ×´ÎÌá³öÁËÓÃÀîȺºÍÀî´úÊý×öƽ̨ȺʵÏÖÃÜÂë·½°¸µÄ¹Ûµã¡£ÎÒÃÇÊ×ÏÈÀûÓþØÕóÀîȺµÄÖ¸ÊýÓ³É䣬½áºÏ·Ç½»»»ÈºµÄ·Ö½âÎÊÌ⣬Àà±ÈÓÚÓÐÏÞÓòÉϵÄÀëÉ¢¶ÔÊýÎÊÌâÌá³ö»ùÓÚÀîȺºÍÀî´úÊýµÄÓйØÄÑÌâ(¼ÙÉè)£»È»ºó»ùÓÚÉÏÊöÄÑÌâ¼ÙÉ裬ÎÒÃÇÉè¼ÆÁËa new miniature CCA2 public key encryption scheme in the random oracle models.


%\noindent \textbf{Mathematics Subject Classification (2010)}  94A60 $\cdot$ 11T71 $\cdot$ 14G50 $\cdot$ 20G40 $\cdot$ 20E28 %$\cdot$  20E32 $\cdot$  20D06 $\cdot$ 05E15 $\cdot$  51A40
\end{abstract}

\section{Introduction}
Currently, most asymmetric cryptographic primitives are based on the perceived intractable problems in number theory, such as the integer factorization problem and discrete logarithm problem. However, due to Shor's and other quantum algorithms \cite{S97,PZ03} for solving the integer factorization problem and discrete logarithm problem, the known public key cryptosystems based on these two assumptions would be broken, when quantum computers become practical. Recent advances in quantum computers shows that the time is coming \cite{news}.
Therefore, it is an imminent work to search for more complex mathematical platforms and to design effective cryptographic schemes, which can resist against quantum attacks.

To deal with the crisis of cryptography in quantum era,  cryptographers has began to pay more attention towards non-commutative cryptography based on non-commutative algebraic structures. One of the outstanding properties of non-commutative cryptography is that it can take the advantage of intractable problems in quantum computing, combinatorial group theory and computational complexity theory to constructing cryptographic platforms. This extension has a profound background and rich connotation. First, from the viewpoint of the platforms, non-commutative cryptography extends the research territory of cryptography. A large number of non-commutative algebraic structures are now waiting to be explored for new public key cryptosystems. Second, due to the ability of resisting against quantum attacks, non-commutative cryptography is expected to achieve a higher strength. It is well known that non-commutative algebraic structures can increase the hardness of some mathematical problems significantly. For instance, we already know that how to design efficient quantum algorithms for solving hidden subgroup problems in any abelian group, but we are still unable to construct efficient algorithms for dealing hidden subgroup problem in non-abelian groups \cite{R06}.

Most of cryptosystems in non-commutative cryptography are derived from combinatorial group theory, but they are mainly theoretical or have certain limitations in wider and general practice. This is perhaps due to the lack of appropriate description of group elements and operations or the difficulty of implementing cryptosystems in practical domains. The non-abelian group (Lie group) used in this paper is quite simple with clear description of group elements and operations and it is easy to implemented.

\subsection{Our Motivations and Contributions}
Lie groups have important applications in many branches of physics and mathematics such as mathematical analysis, differential geometry, topology and quantum mechanics. Lie theory originated from Lie's idea that extends the Galois theory for algebraic equations to the differential equations \cite{H03}. From its beginning, Lie theory was inextricably linked with the developments of algebra, analysis and geometry. As the important measure of algebraic properties of Lie groups, Lie algebras play an indispensable tool while studying matrix Lie groups. On the one hand, Lie algebras are simpler than Lie groups. On the other hand, the Lie algebra of a matrix Lie group contains much information about that group.

In Lie theory, matrix Lie groups are important among the types of Lie groups and have classical matrix forms with their Lie algebras. After exploring cryptographic aspects of Lie theory, we extracted an interesting discovery: the exponential mapping between Lie groups and Lie algebras can be viewed as a non-abelian analog of exponent operation in finite fields. While being different from exponent operation in finite fields, the exponential mapping is the usual power series of Lie algebras, and the image set is indeed Lie groups. Besides, there are different intractable assumptions between them: exponent operation in finite fields is based on DLP in finite fields; the exponential mapping is based on solving root problem of high degree polynomial equation in one variate, which can be viewed as a variant version on matrices. Currently, there are no direct formulas to solve this problem rather than degrading the degree of the equation step by step. When the variant is matrix, the complexity increases rapidly. Therefore, combing cryptographic aspects of the exponential mapping, we probe some cryptographic applications based on Lie theory.

In this paper, we come up with a series of intractable assumptions based on the exponential mapping in Lie theory, including the non-abelian factoring assumption and non-abelian inserting assumption. Subsequently, we propose a CCA secure public key encryption scheme by using the FO technique \cite{PKE:FO99b}. We also give the security proof in the random oracle based on the new assumption.

\subsection{Related Works}
It is always the most important thing to study the underlying intractable hypothesis of mathematical problems for cryptographic primitives. Regarding the
non-commutative cryptography, this kind of study started from 1980's when the
difficult problems in group theory were applied into cryptography. In 1984, Wagner et al.\cite{WM84} designed a public key cryptosystem based on undecidable word problem in groups and semigroups. In 2000, Ko et al. \cite{KLCHKP00} developed braid group cryptography based on the intractable assumption of conjugate search problem in braid group. In 2004, Eick and Kahrobaei \cite{EK04} devised a new cryptosystem based on the polycyclic group. In 2005, Shpilrain and Ushakov \cite{SU05} put forward a new public key cryptosystem by using Thomsen group. Since 2011, Kahrobaei et.al \cite{KKS13,KKS14,KA13,HKK13,MSU11} devised several new key exchange schemes and public key encryption schemes based on group ring matrix, corresponding intractable assumptions are reported to be DLP and FP in group ring matrix, respectively. Unfortunately,
most of the above cryptographic schemes are not secure\cite{MSU11}

At the same time, a type of cryptosystems based on the intractable assumption in non-abelian group---group factorization problem (GFP) has gradually become a typical representative of non-commutative cryptography and achieved rapid development in recent thirty years. The first work in this type of cryptosystems is the symmetric cryptosystem|PGM based on a special factorization basis in finite permutation groups|logarithmic signature (LS) proposed by Magliveras in 1986\cite{M86}. The algebraic properties of PGM were studied more deeply in \cite{MM89,MM90,MM92,CV06}, algebraic properties of PGM were discussed in detail. In 2002, In 2002, Magliveras et al. \cite{MST02} put forward a trapdoor permutation function and two public key cryptosystems MST1 and MST2 by employing LS in finite non-abelian groups. In 2009, Magliveras et al. \cite{LMTW09} devised a new public key cryptographic system---$MST_3$ based on random covers and LS in finite non-abelian groups. Meanwhile, Magliveras et al. proposed a practical platform--Suzuki 2-group for the first time \cite{H63} and devised MST cryptosystems into practice. However, some of the weaknesses are found in MST
series cryptosystems \cite{MSTZ08,BCM09,VPD10,ST10}. In 2008, Magliveras et al. \cite{MSTZ08} provided a comprehensive analysis of $MST_3$ cryptosystem and stated that transitive LS is not suitable for $MST_3$ cryptosystem. In 2009, Blackburn et al. \cite{BCM09} pointed out that amalgamated LS is also not a reasonable choice for MST cryptosystems. In 2010, Vasco et al. \cite{VPD10} presented a more profound analysis of $MST_3$ and showed that the intractability assumption GFP doesn't always hold for random cover of group $G$. The authors also discussed that MST3 cryptosystem cannot achieve
one-wayness in chosen plaintext attack model, let alone the indistinguishability
against adaptive chosen ciphertext attacks. Therefore, in 2010, Svaba et al. \cite{ST10} constructed a more secure cryptosystem $eMST_3$ by employing a secret homomorphic map. Moreover, the authors analyzed all of the published references about attacking MST cryptosystems and developed a set of weak key test tool for $eMST_3$ cryptosystem. It was claimed that bad LSs can be replaced by employing presented tool. But until now, there is no valid evidence showing that this method is reasonable and effective.

Though there are many non-commutative cryptosystems proposed till now,
none of them are proven secure against chosen ciphertext attacks.

\subsection{Paper Organization}
The remaining paper is organized as follows.
In Section~\ref{sec:definition}, we will review the related results in Lie groups, and propose our new assumptions. In Section \ref{sec:proposal}, we present our CCA public key encryption in Lie groups with along its security analysis and efficiency analysis. At last, we conclude
the paper in Section~\ref{sec::conclusions}.

\section{Definitions}\label{sec:definition}
In this section, we will review the definitions related to Lie groups, and propose the non-abelian discrete logarithm (NAF) problem and non-abelian inserting (NAI) problem, as well as the hardness analysis. For clarity, we would like to introduce the notations used in this paper.
\begin{table}
  \centering
  \caption{Notations used in this paper.}\label{table::notations}
  \begin{tabular}{|l|l|}
    \hline
    $\mathbb{R}$ & set of real numbers \\
    $\mathbb{C}$ & set of complex numbers \\
    $\mathbb{Z}$  & set of integers\\
     $M_n(\mathbb{C})$ & set of $n\times n$ complex matrices\\
    $GL_n(\mathbb{C})$ & set of all invertible $n\times n$ matrices with complex entries\\
    $p$ & large prime number\\
    $M_n(p)$ & set of $n\times n$ matrices with entries in $\mathbb{Z}_p$ \\
    $GL_n(p)$ & set of all invertible $n\times n$ matrices with entries in $\mathbb{Z}_p$\\
    $\exp$ & natural logrithm\\
    \hline
  \end{tabular}
\end{table}

\subsection{Matrix Exponential and One-Parameter Subgroup}

In this section, we review several classical conclusions in Lie theory including
matrix exponential and one-parameter subgroup. Actually, we directly copy the
results from \cite{H03}.

\begin{definition}[Matrix Exponential]\cite{H03}
Let $X\in M_n(\mathbb{C})$ be an $n\times n$ complex matrix, then the matrix exponential of $X$ is defined as the usual power series $\exp^X=\sum_{m=0}^{\infty}\frac{X^m}{m!}$. In case when $X$ is a nilpotent matrix, $\exp^X=\sum_{m=0}^{\ell}\frac{X^m}{m!}$, where~$\ell$ is the nilpotent index of $X$.
\end{definition}

It is easy to see that $M_n(\mathbb{C})$ along with the multiplication operation construct a semigroup.

\begin{proposition}\cite{H03}
Let $X$ and $Y$ be arbitrary $n\times n$ matrices. Then, we have the following:

\begin{enumerate}
  \item $\exp^0=I_n$.
  \item $\exp^X$ is invertible and $(\exp^X)^{-1}=\exp^{-X}$.
  \item $\exp^{(\alpha+\beta)X}=\exp^{\alpha X}\cdot \exp^{\beta X}$ for all $\alpha$ and $\beta$ in $\mathbb{C}$.
  \item If $XY=YX$, then $\exp^{X+Y}=\exp^X\cdot\exp^Y=\exp^Y\cdot\exp^X$.

\end{enumerate}

\end{proposition}

Item 3 shows that for an arbitrary matrix $X$, the power series $\exp^X$ is an invertible matrix and belongs to $GL_n(\mathbb{C})$. Item 4 describes that the commutativity of $\exp^X$ and $\exp^Y$ depends on the matrices $X$ and $Y$.

\begin{definition}[One-Parameter Subgroup]\cite{H03}
A function $F:\mathbb{R}\rightarrow GL_n(\mathbb{C})$ is called a one-parameter subgroup of ~$GL_n(\mathbb{C})$ if
\begin{enumerate}
  \item $F$ is continuous;
  \item $F(0)=I_n$;
  \item $F(t+s)=F(t)F(s)$ for all $t,s\in\mathbb{R}$.
\end{enumerate}

\end{definition}

\begin{property}\label{tb:1}\cite{H03}
If $F$ is a one-parameter subgroup of $GL_n(\mathbb{C})$, then there exists a unique $n\times n$ complex matrix $X\in M_n(\mathbb{C})$ such that
\begin{center}
  $F(t)=\exp^{tX}$
\end{center}
\end{property}

%\textbf{×¢Ò⣺}
%º¯Êý~$F$ µÄÐÔÖÊÔÚÀîȺÊé29Ò³ÃüÌâ2.3ÓнéÉÜ¡£ÔÚ¶¨Àí\ref{tb:1}ÖУ¬

In Lie theory, $\exp^{tX}$ is the exponential mapping from a Lie algebra $X$ to its Lie group. Meanwhile, when $X$ is given, $F(t)=\exp^{tX}\in GL_n(\mathbb{C})$ is an injection and a one-way function. Specially, the injection property is implied by Proposition 1 (items 1,2,3), and the one-wayness is due to the intractable assumptions of solving high degree root problem of polynomial equation in one variate \cite{M75,MS97,J12}.

\subsection{Non-abelian Factoring Problem and Non-abelian Inserting Problem}
By using the results reviewed above, we can propose two hard problems: non-abelian factoring (NAF) problem and non-abelian inserting (NAI) problem.

\begin{definition}[Non-abelian Factoring(NAF) Problem ]
Let $\mathbb{M}=M_n(p)$ be a semigroup with respect to multiplication operation, and $\mathbb{G}=GL_n(p)$ the general linear group with respect to multiplication operation. Let $R, T\in \mathbb{M}$ $(R\neq T)$ be two random nilpotent matrices. The factoring problem with respect to $\mathbb{G}, R, T$, denoted by $\mathtt{NAF}_{\exp^R,\exp^T}^\mathbb{G}$, is to factor the given product $\exp^{xR}\cdot\exp^{yT}\in \mathbb{G}$ into a pair $(\exp^{xR},\exp^{yT})\in \mathbb{G}^2$.
\end{definition}

Now, let's analyze the hardness of the NAF problem. Firstly, it is easy to see that there are many forms for $A= \exp^{xR}\cdot\exp^{yT}$. For instance, $A=BC=B'C'$. Secondly, from Proposition 1, we get that the map $(x,y)\mapsto \exp^{xR}\cdot\exp^{yT}$ is an injection with respect to $R$ and $T$. Hence, it is with probability $1/|\mathbb{G}|\approx1/p^{n^2}$ at most to find a specific pair $(x,y)$ satisfing the maps $x\mapsto \exp^{xR}$, $y\mapsto \exp^{yT}$ and $\exp^{xR}\cdot\exp^{yT}$ simultaneously. Note that $|\mathbb{G}|<|\mathbb{M}|=p^{n^2}$ and $|\mathbb{G}|\approx|\mathbb{M}|=p^{n^2}$ when $p$ is large enough. As a result, we believe that the NAF problem is hard when $|\mathbb{G}|$ is large.

Furthermore, if $R$ and $T$ are noncommutative, so from Proposition 1 (items 1, 2 and 3), we conclude that $\exp^{xR}$ and $\exp^{yT}$ are non-commutative. In this paper, we always assume that $R$ and $T$ are non-commutative, $n\geq 5$ and $p$ is large enough.

It is quite interesting that solving the problem that given $\exp^{tX}\in \mathbb{G}$ and $X\in \mathbb{M}$ to compute $t$ does not help to solve  the NAF problem. It is because that once $R\neq T$, there does not exist any operation between $\exp^{xR}$ and $\exp^{yT}$ or between $\exp^R$ and $\exp^T$.

\begin{definition}[Non-abelian Inserting (NAI) Problem]
Let $\mathbb{M}=M_n(p)$ be a semigroup with respect to multiplication operation, and $\mathbb{G}=GL_n(p)$ the general linear group with respect to multiplication operation. Let $R, T\in \mathbb{M}$ $(R\neq T)$ be two random nilpotent matrices.The non-abelian inserting (NAI) problem with respect to $\mathbb{G}, R, T$, denoted by $\mathtt{NAI}_{\exp^R,\exp^T}^\mathbb{G}$, is to recover $\exp^{(a+c)R}\cdot\exp^{(b+d)T}$ from the given random pair ($\exp^{aR}\cdot\exp^{bT}$, $\exp^{cR}\cdot\exp^{dT}$)$\in \mathbb{G}^2$.
\end{definition}

It is easy to see that if the NAF problem is easy, then the NAI problem can be also solved. In particular, the adversary can use the solution of the NAF problem to get $\exp^{aR}$ and $\exp^{bT}$ with input $\exp^{aR}\cdot \exp^{bT}$. After that, the adversary can obtain the NAI solution $\exp^{aR}\cdot \exp^{cR}\cdot\exp^{dT} \cdot \exp^{bT}$.

Actually, due to the  non-commutability, the best solution for the NAI problem is to split one item of the NAI input into two parts, and then combine all of them together. It looks like one item of the NAI input is inserted into the other item. Hence, the name is obtained.

\section{Proposed Public Key Encryption Scheme in Lie Groups}\label{sec:proposal}
In this section, we will propose a new public key encryption scheme in Lie groups by using the FO technique \cite{PKE:FO99b}. In particular, our proposal is proven-secure against chosen ciphertext attacks in the random oracle model assuming the inserting problem is hard in the underlying Lie group.

\subsection{Description of the Proposal}
There exist three algorithms in our proposal: key pair generation algorithm $\KeyGen$, encryption algorithm $\Enc$, and decryption algorithm $\Dec$. The details are as follows.

\begin{description}
  \item[$\KeyGen(\kappa)$:] It takes the security parameters $\kappa_1,\kappa_2,\kappa_3,$ as input, it outputs a public key $pk=(\mathbb{M},\mathbb{G},S,T,\Delta,H_1,H_2,H_3)$, and the corresponding private key $sk=(\exp^{x\cdot S},\exp^{y\cdot T})$. The key pair satisfies the following requirements.
      \begin{itemize}
        \item $\mathbb{M}=M_n(p)$ is a semigroup with respect to multiplication operations.
        \item $\mathbb{G}=GL_n(p)$ is a non-abelian matrix Lie group with rank $n(n\geq 5)$.
        \item $p$ is a large prime number with $p=\Theta(2^{\kappa_1})$, and $|\mathbb{G}|=\Theta(p^{n^2})=\Theta(2^{n^2{\kappa_1}})$.
        \item $R, T \in \mathbb{M}$ are two random nilpotent matrices, and $\Delta=\exp^{s\cdot S}\cdot \exp^{t\cdot T}$, where $s\in \{0,1\}^{\kappa_3}$ and $t\in \{0,1\}^{\kappa_4}$ are random numbers.
        \item $H_1,H_2,H_3$ are three cryptographically secure hash functions: $H_1:\{0,1\}^{\kappa_2+\ell}\rightarrow \{0,1\}^{\kappa_3+\kappa_4}$, $H_2:\mathbb{G}\rightarrow \{0,1\}^{\kappa_2}$, and $H_3:\{0,1\}^{\kappa_2}\rightarrow \{0,1\}^{\ell}$, where $\ell$ is the bit length of the message.
      \end{itemize}
      At last, $s,t$ should be securely destroyed.
  \item[$\Enc(pk,m)$:] It takes a public key $pk=(\mathbb{M},\mathbb{G},S,T,\Delta,H_1,H_2,H_3)$ and a message $m\in \{0,1\}^{\ell}$ as input, it outputs the corresponding ciphertext $C=(C_1,C_2,C_3)$ by doing the following steps.
      \begin{itemize}
        \item Choose randomly a number $\sigma$ from $\{0,1\}^{\kappa_2}$.
        \item Compute $r_s||r_t=H_1(\sigma||m)$.
        \item Compute $C_1=H_2(\exp^{r_s\cdot S}\cdot \Delta \cdot \exp^{r_t\cdot T})\oplus \sigma$.
        \item Compute $C_2=\exp^{r_s\cdot S}\cdot \exp^{r_t\cdot T}$.
        \item Compute $C_3=H_3(\sigma)\oplus m$.
      \end{itemize}
  \item[$\Dec(sk,C)$:] It takes a private key $sk=(\exp^{s\cdot S},\exp^{t\cdot T})$ and a ciphertext $C=(C_1,C_2,C_3)$ as input, it outputs the corresponding message as follows.
      \begin{itemize}
        \item Compute $\sigma'=C_1\oplus H_2(\exp^{s\cdot S}\cdot C_2\cdot \exp^{t\cdot T})$.
        \item Compute $m'=C_3\oplus H_3(\sigma)$.
        \item Compute $r'_s||r'_t=H_1(\sigma'||m')$.
        \item Check whether both of $C_1=H_2(\exp^{r'_s\cdot S}\cdot \Delta \cdot \exp^{r'_t\cdot T})\oplus \sigma'$ and $C_2=\exp^{r'_s\cdot S}\cdot \exp^{r'_t\cdot T}$ hold. If they both hold, set $m=m'$; otherwise, set $m=\bot$.
        \item Output $m$.
      \end{itemize}
\end{description}

\paragraph{Correctness of the Proposal.}
The correctness of the proposal can be easily obtained by the following equalities.
\begin{eqnarray*}
   && \exp^{r_s\cdot S}\cdot \Delta \cdot \exp^{r_t\cdot T} \\
   &=& \exp^{r_s\cdot S}\cdot \exp^{t\cdot S}\cdot \exp^{t\cdot T} \cdot \exp^{r_t\cdot T} \\
   &=& \exp^{s\cdot S}\cdot \exp^{r_t\cdot S}\cdot \exp^{r_t\cdot T} \cdot \exp^{t\cdot T}\\
   &=&  \exp^{s\cdot S}\cdot C_2 \cdot \exp^{t\cdot T}
\end{eqnarray*}

\subsection{Security Analysis of the Proposal}
By the techniques used in \cite{PKE:FO99b}, we can prove that our proposal is secure against the chosen chiphertext attacks in the random oracle model assuming that the inserting problem in the Lie group is hard.

\begin{theorem}
The proposal is secure against the chosen chiphertext attacks in the random oracle model based on the NAI assumption in the Lie group.
\end{theorem}

\begin{proof}
If there exists an adversary $\mathcal{A}$ can break the CCA security of the proposal, then we can build another algorithm $\mathcal{B}$ solving the non-abelian inserting problem in the Lie group. That is, given $\Delta_1=\exp^{s_1\cdot S}\cdot \exp^{t_1\cdot T}\in\mathbb{G}$, $\Delta_2=\exp^{s_2\cdot S}\cdot \exp^{t_2\cdot T}\in \mathbb{G}$, and $S,T\in \mathbb{M}$, it aims to output $\Delta=\exp^{(s_1+s_2)\cdot S}\cdot \exp^{(t_1+t_2)\cdot T}$. The details are as follows.

\begin{description}
  \item[Setup:] $\mathcal{B}$ sets the public values $S,T,\Delta$ as $S,T,\Delta_1=\exp^{s_1\cdot S}\cdot \exp^{t_1\cdot T}$, respectively. Clearly, $\mathcal{B}$ has no idea about the corresponding private key $sk=(\exp^{s_1\cdot S},\exp^{t_1\cdot T})$.
  \item[Phase 1:] $\mathcal{B}$ builds the following oracles.
  \begin{itemize}
    \item Random Oracle $\mathcal{O}_{H_1}$: $\A$ sends $\sigma||m\in \{0,1\}^{\kappa_2+\ell}$ to this oracle, $\B$ firstly searches whether $(\sigma||m, \alpha)$ exists in Table $T_{H_1}$ that is empty at the beginning. If it exists, $\B$ returns $\alpha$ to $\A$; otherwise, $\B$ chooses a random number $\alpha$ from $\{0,1\}^{\kappa_3+\kappa_4}$, records $(\sigma||m, \alpha)$  into Table $T_{H_1}$, and sends $\alpha$ to $\A$.
    \item Random Oracle $\mathcal{O}_{H_2}$: $\A$ sends $R\in \mathbb{G}$ to this oracle, $\B$ firstly searches whether $(R, \beta)$ exists in Table $T_{H_2}$ that is empty at the beginning. If it exists, $\B$ returns $\beta$ to $\A$; otherwise, $\B$ chooses a random number $\beta$ from $\{0,1\}^{\kappa_2}$, records $(R, \beta)$  into Table $T_{H_2}$, and sends $\beta$ to $\A$.
    \item Random Oracle $\mathcal{O}_{H_3}$: $\A$ sends $\sigma\in \{0,1\}^{\kappa_2}$ to this oracle, $\B$ firstly searches whether $(\sigma, \gamma)$ exists in Table $T_{H_2}$ that is empty at the beginning. If it exists, $\B$ returns $\gamma$ to $\A$; otherwise, $\B$ chooses a random number $\gamma$ from $\{0,1\}^{\ell}$, records $(\sigma, \gamma)$  into Table $T_{H_2}$, and sends $\gamma$ to $\A$.
    \item Decryption Oracle $\mathcal{O}_{dec}$: $\A$ sends a ciphertext $C=(C_1,C_2,C_3)\in \{0,1\}^{\kappa_2}\times \mathbb{G} \times \{0,1\}^{\ell}$ to this oracle, $\B$ firstly searches  $(\sigma,m,\alpha,\beta,\gamma)$ in tables $T_{H_1}$, $T_{H_2}$ and $T_{H_3}$, where $\alpha_s||\alpha_t=\alpha=H_1(\sigma||m)$, $C_1=\beta\oplus \sigma$, $C_2=\exp^{\alpha_x\cdot S}\cdot \exp^{\alpha_t\cdot T}$, and $C_3=\gamma\oplus m$. It it exists, $\B$ sends $m$ to $\A$; otherwise, $\B$ sends $\bot$ to $\A$.
  \end{itemize}
  \item[Challenge:] $\A$ sends $\B$ two messages $m_0,m_1\in \{0,1\}^{\ell}$ with equal bit length. $\B$ computes $C^*=(C^*_1,C^*_2,C^*_3)$ as follows.
      \begin{itemize}
        \item Choose random $\sigma^*,\beta^*$ from $\{0,1\}^{\kappa_2}$, and compute $C^*_1=\sigma^*\oplus \beta^*$.
        \item Set $C_2^*=\Delta_2$.
        \item Compute $C_3^*=H_3(\sigma)\oplus m_b$, where $b$ is a random number from $\{0,1\}$.
      \end{itemize}
      At last, $\B$ sends $C^*$ to $\A$ as the challenge ciphertext.
  \item[Phase 2:] It is almost the same as Phase 1, except that $\A$ cannot directly send $C^*$ to the decryption oracle $\mathcal{O}_{dec}$.
  \item[Guess:] $\A$ outputs the guess $b'$ on $b$. $\B$ randomly chooses $R$ from Table $T_{H_2}$, and sets $\Delta$ as $R$. If $\A$ can output a correct guess, then $R$ is the right $\Delta$ with probability $1/q_{H_2}$ at least, where $q_{H_2}$ is the maximum number of queries to the random oracle $\mathcal{O}_{H_2}$ by $\A$.
\end{description}
Similar with the analysis in \cite{PKE:FO99b}, we can conclude that our proposal is secure against chosen ciphertext attacks based on the NAI assumption. \qed
\end{proof}

\subsection{Quantum Algorithm Attacks}
Since the publication of Shor's quantum algorithm for solving IFP and DLP \cite{S97}, many mathematicians devote into developing secure public key cryptosystems based on non-abelian algebra.
%BKT¡¯ proposal provides an important analogy from factoring the abelian ring $\mathbb{Z}_n$ (where $n$ is a product of two secure primes) to factoring sub-structures over some non-abelian algebra. Although BKT¡¯s original design and sample implementations are insecure (Section 4), BKT¡¯s analogy is laudable.
It is unclear that how to use Shor¡¯s quantum algorithm to break the intractability assumption of the $\mathtt{NAI}_{\exp^R,\exp^T}$ problem.

 Recall that Shor's algorithm \cite{S97} consists of two parts: a quantum algorithm to solve the order-finding problem over $\mathbb{Z}_n^*$ and a classical reduction of factoring $n$ to the problem of order finding. Now, let us show that even if a quantum algorithm for solving the order-finding problem over a non-abelian group $\mathbb{G}$ is at hand, at present we still have no reductions, either classical or quantum for underlying problem. In fact, the exponential mapping is completely different from exponential operation in finite fields. Moreover, since $R$ and $T$ are both nilpotent matrices, there is no order of a nilpotent matrix. Hence, Shor's algorithm cannot work for this case.

 On the other hand, in order to obtain the pair $(\exp^{xR},\exp^{yT})$, we have to factorize $\exp^{xR}\cdot\exp^{yT}\in \mathbb{G}$. But until now, there is no efficient classical algorithms or quantum algorithms for factoring a general matrix into two specific matrices.

 Consequently, our scheme is secure against known classical and quantum algorithms.

\subsection{Efficiency Analysis}
In this section, we would like to analyze the efficiency of our proposal and how to choose the security parameters. In particular, we have the followings.

\begin{enumerate}
  \item
  Key generation algorithm requires two exponential mappings of two nilpotent matrices $S$ and $T$, and the core parameters of the public key (pk) and the secret key (sk) are the triple ($S, T, \exp^{sS}\cdot\exp^{tT}$) and the pair $(\exp^{sS},\exp^{tT})$, respectively. They are $3|p^{n^2}|$ and $2|p^{n^2}|$ bit length respectively. Here, we ignore the part of the parameters to describe $\mathbb{M},\mathbb{G},H_1,H_2,H_3$.

  \item
  Encryption algorithm requires two exponential mappings to compute $\exp^{r_sS}$ and $\exp^{r_tT}$ and additional three multiplications to get the final ciphertext. Similarly, the cost for evaluating $H_1$, $H_2$ and $H_3$ is ignored without loss of generality. The bit length of one ciphertext is $\kappa_2+|p^{n^2}|+\ell$.

  \item
  Decryption algorithm does not need any exponential mappings but only two multiplications to get the message, while it needs two exponential mappings and three multiplications to check the validity of the ciphertext. The cost of evaluating hash functions are still ignored.

  \item According the results in Section 2, the ranges of $s,r_s,t,r_t$ could be extended to $\mathbb{Z}$. In order to easy implementation, we set the ranges as $\{0,1\}^{\kappa_3}$ and $\{0,1\}^{\kappa_4}$ in the description of our proposal. On the other hand, $\kappa_3$ and $\kappa_3$ should be large enough to resist against the brute force attack. Recall the analysis of the NAF problem, the hardness is related to $|\mathbb{G}|\approx p^{n^2}$. Hence, $\kappa_1= |p|$ and $n$ should be large enough to make |G| large. At last, $\kappa_2$ could be set as that in \cite{PKE:FO99b}.

%
%To the best of our knowledge, the best algorithm for the NAI problem is brute force. Hence, $\kappa_1\geq 160/25$, and $\kappa_2=\kappa_3=\kappa_4\geq 160$ are enough for 80 bit security.
\end{enumerate}

\section{Conclusion}\label{sec::conclusions}

The invention of Shor's quantum algorithm for solving integer factorization problem and discrete logarithm problem casts distrust on many public key cryptosystems used today. This urges us to develop secure public key cryptosystems based on variety platforms, such as non-abelian algebra. In this paper, we at first presented two new intractable assumptions by using the exponential mapping in Lie group. Subsequently, we proposed a new public key encryption schemes based on Lie groups and Lie algebras. Our proposals are proved to be CCA secure in the random oracle model.

\section*{Acknowledgements}
This work is partially supported by the National Natural Science Foundation of China (NSFC) (Nos.61502048, 61370194) and the NSFC A3 Foresight Program (No.61411146001).

%
%\bibliographystyle{plain}%plain, unsrt,alpha
%\bibliography{bib}

\end{document}